\documentclass[11pt]{article}

\usepackage[margin=1in]{geometry}
\usepackage[T1]{fontenc}
\usepackage{lmodern}
\usepackage{amsmath,amssymb,amsthm,mathtools}
\usepackage{algorithm}
\usepackage{algpseudocode}
\usepackage{thmtools}
\usepackage{booktabs}
\usepackage{enumitem}
\usepackage{xcolor}
\usepackage{tikz}
\usetikzlibrary{arrows.meta,calc,positioning}
\usepackage{hyperref}
\usepackage{microtype}
\usepackage{verbatim}

\hypersetup{
  colorlinks=true,
  linkcolor=blue!50!black,
  citecolor=blue!50!black,
  urlcolor=blue!50!black
}

\newtheorem{theorem}{Theorem}[section]
\newtheorem{lemma}[theorem]{Lemma}

\newcommand{\OPT}{\operatorname{OPT}}
\newcommand{\LCS}{\operatorname{LCS}}
\newcommand{\cost}{\operatorname{cost}}
\newcommand{\rad}{\operatorname{rad}}
\newcommand{\val}{\operatorname{val}}

\newcommand{\parity}{\mathrm{Parity}}
\newcommand{\var}{\mathrm{var}}
\newcommand{\pol}{\mathrm{pol}}

\definecolor{greenlantern}{rgb}{0.0, 0.5, 0.0}

\definecolor{red}{rgb}{1.0, 0.0, 0.0}

\title{Hardness of Approximation of Rank Aggregation on Ulam Metric}

\author{%
Sk Ruhul Azgor%
  \thanks{Pennsylvania State University.
    Email: \texttt{sfa6135@psu.edu}
  }
\and
Diptarka Chakraborty%
\thanks{National University of Singapore.
  Work partially supported by an MoE AcRF Tier 1 grant (T1 251RES2303) and a Google South \& South-East Asia Research Award.
    Email: \texttt{diptarka@nus.edu.sg}
  }
\and
Le Van Cuong%
\thanks{National University of Singapore.
    Email: \texttt{e1583414@u.nus.edu}
  }
\and
Debarati Das%
\thanks{Pennsylvania State University.
Work supported in part by NSF grant 2337832.
        Email: \texttt{debaratix710@gmail.com}
}
\and
Tien Long Nguyen%
  \thanks{Pennsylvania State University.
    Email: \texttt{tfn5179@psu.edu}
  }
}

\date{}

\begin{document}
\maketitle

\begin{abstract}
    We study the approximability of rank aggregation under the Ulam metric. In the \emph{Ulam median} problem, the goal is to find a permutation minimizing the sum of its Ulam distances to the input permutations, while in the \emph{Ulam center} problem the objective is to minimize the maximum such distance. Both problems are known to be NP-hard, but no explicit approximation hardness was previously known. We prove that, for every $\varepsilon>0$, it is NP-hard to approximate either Ulam median or Ulam center within a factor of $51/50-\varepsilon$, even when the input consists of only four permutations. We further show that unless P = NP, neither problem admits a polynomial-time additive approximation scheme. The hardness result for Ulam median is established via a reduction from MAX-E3-LIN-2. The corresponding hardness for Ulam center is then obtained through a reduction from Ulam median.
\end{abstract}

\section{Introduction}\label{sec:introduction}

Aggregating potentially inconsistent information from multiple sources is a fundamental challenge across social choice, information retrieval, and data analysis. When sources such as voters, search engines, or learned agents provide complete orderings over a set of alternatives, the problem is known as \emph{rank aggregation}. The goal is to compute a consensus ranking that optimizes a specific objective function. Because different individuals or retrieval systems rarely agree on every pair of alternatives, input rankings inherently conflict. This dynamic naturally arises in diverse applications, including voting, web search, database systems, similarity search, and classification~\cite{brandt2016handbook, DKNS01, Harman92a, FKS03, ACN08}.

Two of the most well-studied optimization objectives in this domain are the \emph{median} and \emph{center} variants. The median variant seeks a consensus ranking that minimizes the sum of distances to the input rankings, while the center variant minimizes the maximum distance to any input. In both cases, the optimal consensus may be any valid permutation, not strictly one of the original inputs. Ultimately, the computational complexity of solving either variant depends heavily on the underlying distance metric used to compare permutations.

In this paper, we study the rank aggregation problem under the \emph{Ulam distance}. The Ulam distance between two rankings is the minimum number of items/candidates that need to be moved to transform one ranking into the other, where a move removes an item and reinserts it elsewhere.  Equivalently, for permutations
of $n$ items, it equals $n$ minus the length of their \emph{longest common
subsequence} (LCS).  This distance metric is compelling primarily for two reasons.  First, the Ulam distance is exactly one-half of the insertion-deletion distance, a natural variant of the standard edit distance, when both strings are permutations of the same alphabet. Thus, it captures much of the alignment structure that makes edit-distance problems challenging, while avoiding the complications caused by repeated symbols. This connection suggests that insights for Ulam aggregation may also be useful for more general multiple sequence alignment problems, which have numerous applications in computational biology~\cite{gusfield1997, pevzner2000computational}, DNA storage systems~\cite {GBCDLSB13, RMRAJY17}, speech recognition~\cite{kohonen1985median}, and classification~\cite{martinez2000use}. Second, a move operation naturally models a displaced item while preserving the relative order of a large unaffected subsequence, a common form of error in rankings. Accordingly, the Ulam distance arises naturally in applications involving rankings, including social choice theory~\cite{brandt2016handbook} and information retrieval~\cite{Harman92a}, and has been extensively studied in the algorithmic literature in various contexts~\cite{CMS01, CK06, AN10, NSS17, BS19}.

Both the median and center variants of the rank aggregation problem under the Ulam distance are already known to be NP-hard~\cite{fischer2025}. Recently, through a reduction from 3SAT, the Ulam median problem was shown to be NP-hard even for four input permutations~\cite{Habib26}. On the algorithmic side, returning the best input gives a folklore $2$-approximation to the median by the triangle inequality. For the center, returning any input gives the same guarantee. Beating this generic bound for the Ulam median long remained a computational challenge, and was first achieved in~\cite{chakraborty2021approximating}, which proposed a $(2-\delta)$-approximation for a tiny constant $\delta>0$. Since then, several new algorithmic frameworks~\cite{CDK23, CDN26}, along with refinements~\cite{jaiswal2025robust}, have been proposed. Despite these attempts, the best-known approximation factor is $1.968$~\cite{CDN26}. On the center variant, the scenario is much more elusive, with no better-than-$2$-factor algorithm known, and a slightly better than $3/2$-approximation known only when the number of inputs is constant~\cite{CGJ21}. This behavior contrasts sharply with rank aggregation under the \emph{Kendall-tau distance} -- another fundamental distance metric defined over permutations. Although the corresponding median problem is NP-hard~\cite{DKNS01}, even for three input rankings~\cite{peters2026kemeny, madarasi2026complexity}, it admits a polynomial-time approximation scheme (PTAS)~\cite{kenyon2007rank}. 

It is therefore natural to ask whether Ulam median also admits a PTAS, or even just a polynomial-time additive approximation scheme. We show that neither is possible unless P=NP.

\subsection{Our results}

We begin by proving an explicit multiplicative inapproximability bound for the median objective. In particular, the hardness holds even when the instance consists of only four input permutations, thereby ruling out a PTAS in this highly restricted setting.

\begin{restatable}{theorem}{FourInputHardnessTheorem}
\label{thm:four-input-hardness}
Unless P=NP, for any $\varepsilon > 0$, there is no polynomial-time algorithm that approximates the Ulam median within a factor of $51/50 - \varepsilon$, even when the input consists of four permutations.
\end{restatable}

Furthermore, the reduction used to prove the preceding result produces a gap linear in the total input length \(mn\). Consequently, we can rule out a polynomial-time additive approximation scheme even for instances consisting of only four input permutations.

\begin{restatable}{theorem}{MedianAdditiveHardness}
\label{thm:median-additive-hardness}
Unless P=NP, for any $0 < \varepsilon \leq 1/300$, there is no polynomial-time algorithm that computes an $\varepsilon mn$-additive approximation of Ulam median with $m$ input permutations over an alphabet of size $n$, even when $m = 4$.
\end{restatable}

It is worth emphasizing that four inputs are necessary for such an NP-hardness result: Ulam median can be solved exactly in polynomial time for three input permutations~\cite{chakraborty2021approximating}. Thus, our result establishes a sharp transition from exact tractability with three inputs to constant-factor inapproximability with four. This result is particularly striking compared to the general median-string problem under edit distance, which can be solved in polynomial time for any fixed number of input strings \cite{Sankoff75}. Together with the NP-hardness for four-input Ulam-median established in~\cite{Habib26}, our result highlights the computational significance of the permutation constraint: requiring the median itself to be a permutation can make the problem substantially harder, even when the number of inputs is a fixed constant.

We further show that analogous hardness holds for the Ulam center objective. In particular, we rule out a PTAS even for four input permutations.

\begin{restatable}{theorem}{CenterHardnessTheorem}
\label{thm:center-hardness}
Unless P=NP, for any $\varepsilon > 0$, there is no polynomial-time algorithm that approximates the Ulam center within a factor of $51/50 - \varepsilon$, even when the input consists of four permutations.
\end{restatable}

We further rule out a polynomial-time additive approximation scheme for the Ulam center under the same four-input restriction.

\begin{restatable}{theorem}{CenterAdditiveHardness}
\label{thm:center-additive-hardness}
Unless P=NP, for any $0 < \varepsilon \leq 1/300$, there is no polynomial-time algorithm that computes an $\varepsilon n$-additive approximation of the Ulam center with $m$ input permutations over an alphabet of size $n$, even when $m = 4$.
\end{restatable}

\subsection{Technical overview}
\paragraph{Approximation hardness of Ulam Median.}
Our reduction builds on the reduction introduced to show NP-hardness of the Ulam median problem for four permutations by Habib~\cite{Habib26}: two permutations test local constraints, two enforce consistency among different occurrences of each variable, and a long common anchor block separates these two roles. The anchor block allows every candidate median of the constructed Ulam median instance to be represented by a partition of the non-anchor symbols. We analyze a candidate through its total LCS score, namely, the sum of its LCS lengths with the four input permutations.

In Habib's reduction from 3SAT \cite{Habib26}, soundness is established through an equality-case argument. The clause-witness and variable-consistency gadgets satisfy separate upper bounds, and a median attaining the target score must make both bounds tight. Tightness yields a witness for every clause and a consistent truth value for every variable, which together define a satisfying assignment. This argument, however, does not by itself yield approximation hardness. A near-optimal median need not nearly saturate the two bounds separately: it may make inconsistent choices across different occurrences of a variable and trade a loss in the consistency gadget for a gain in the constraint gadget. An approximation-hardness proof must therefore control the score of every induced partition, including slightly inconsistent partitions, and relate that score quantitatively to the value of a single global assignment. 

Our reduction is from MAX-E3-LIN-2, where we are given a set of parity equations, each over three Boolean variables, and the goal is to maximize the number of satisfied equations. H{\aa}stad~\cite{Hastad01} showed that this problem is NP-hard to approximate within any factor strictly better than $2$. Observe that we can view an assignment as a partition of all variable occurrences into true and false, subject to the requirement that all occurrences of the same variable receive the same value. Within an equation, the true occurrences form a subset of its three variable occurrences, and satisfaction depends only on the parity of this subset. Specifically, an equation is satisfied if an odd or even number of its variables are set to true, depending on the target parity. We encode this behavior in the total LCS score of the four input permutations.

The key ingredient is a parity-sensitive gadget based on two orders of three generic symbols, namely, $P_1 := abc$ and $P_2 := cab$. For any subset $T\subseteq \{a, b, c\}$, its LCS score is the maximum sum of LCS lengths of $P_1$ and $P_2$ to their median, all restricted to $T$. While the score of an individual subset $T$ is not determined solely by its parity, we show that the total score of all subsets with the same parity depends only on this parity, and that there is a fixed unit gap in the scores of the two parity classes. The equation gadget then associates a satisfying assignment with the entire odd-parity class and a violating assignment with the entire even-parity class. Hence, each equation contributes a fixed baseline together with one additional unit exactly when it is satisfied. 

For an arbitrary symbol partition, we decode a global assignment by taking the majority value for each variable on the consistency side. We then compare the LCS score of the given partition with the score of the induced assignment, and show an exact affine relation between the optimum median cost and the maximum number of satisfiable equations. This is precisely the quantitative control over arbitrary partitions required beyond Habib's equality-case argument. Further, it lets us leverage the gap obtained by H{\aa}stad~\cite{Hastad01} to derive hardness-of-approximation results for Ulam median. We present a detailed proof in Section~\ref{sec:median-hardness}.

\paragraph{Median hardness to center hardness.} To transfer the hardness results to Ulam center, we use a cyclic block construction to reduce Ulam median to Ulam center. Given an Ulam median instance with $m$ input permutations, we create $m$ disjoint copies of its alphabet and construct $m$ center inputs. All inputs use the same order of the alphabet copies, while the original permutations are arranged cyclically across them. Thus, every center input contains one copy of each original permutation, and every original permutation appears once in each block across the center instance.

Placing the same median candidate in every block makes its distance to each center input equal to its total distance in the original median instance. Conversely, the common block order lets us rearrange any center candidate blockwise without increasing its radius. Each block then defines a candidate median for the original instance, and at least one of them has median cost no larger than the center radius. The reduction therefore preserves the optimum value exactly and allows a median solution to be extracted from any center solution without incurring any additional loss (see Lemma~\ref{lem:median-to-center}). Consequently, both hardness results for Ulam median transfer to Ulam center. We provide a complete proof in Section~\ref{sec:center-hardness}.

\section{Preliminaries}\label{sec:preliminaries}
\noindent \textbf{Notations. } For $m\in\mathbb N$, we use $[m]$ to denote the set $\{1,\ldots,m\}$. We use the notation $\sqcup$ to denote the union of two disjoint sets. For any set $X$ partitioned into $t$ disjoint subsets $X_1,X_2,\ldots,X_t$, we thus write $X = X_1 \sqcup X_ 2 \sqcup \ldots \sqcup X_t$.

\paragraph{Permutations and Ulam distance.} For an alphabet $\Sigma$, let $\mathcal P(\Sigma)$ denote the set of all permutations of $\Sigma$. For a permutation $\pi\in\mathcal P(\Sigma)$ and a subset $S\subseteq\Sigma$, let $\pi|_S$ denote the restriction of $\pi$ to the symbols in $S$, with their relative order preserved. For two sequences $x$ and $y$, let $\LCS(x,y)$ denote the length of a \emph{longest common subsequence}. For any two $x,y\in\mathcal P(\Sigma)$, their \emph{Ulam distance} is
\[
    d_U(x,y):=|\Sigma|-\LCS(x,y).
\]

We will repeatedly use the following elementary identity, which follows from the triangle inequality for the Ulam distance.

\begin{lemma}\label{lem:two-permutation-score}
Let $\Sigma$ be an alphabet and let $\pi,\sigma\in\mathcal P(\Sigma)$. Then
\[
    \max_{\tau\in\mathcal P(\Sigma)}\bigl(\LCS(\tau,\pi)+\LCS(\tau,\sigma)\bigr)=|\Sigma|+\LCS(\pi,\sigma).
\]
\end{lemma}
\begin{proof}
Let $n:=|\Sigma|$.  For every $\tau\in\mathcal P(\Sigma)$, the triangle
inequality for Ulam distance gives
\[
 d_U(\tau,\pi)+d_U(\tau,\sigma)\geq d_U(\pi,\sigma).
\]
Using the definition of the Ulam distance between any two permutations $\rho,\rho'$, $d_U(\rho,\rho')=n-\LCS(\rho,\rho')$, and rearranging, we obtain
\[
 \LCS(\tau,\pi)+\LCS(\tau,\sigma)
 \leq n+\LCS(\pi,\sigma).
\]
This upper bound is attained by taking $\tau=\pi$, since
$\LCS(\pi,\pi)=n$.  The claimed equality follows.
\end{proof}

\paragraph{Ulam median.} An instance of \emph{Ulam median} is a set $\Pi=\{\pi_1,\ldots,\pi_m\}$ of permutations of a common alphabet $\Sigma$. The problem asks to find a permutation $\sigma\in\mathcal P(\Sigma)$ that minimizes the total distance to input permutations, namely, $\cost_\Pi(\sigma):=\sum_{i=1}^m d_U(\sigma,\pi_i)$. We denote the optimum value of the instance $\Pi$ by
\[
    \OPT_{\mathrm{med}}(\Pi):=\min_{\sigma\in\mathcal P(\Sigma)}\cost_\Pi(\sigma).
\]

\paragraph{Ulam center.} An instance of \emph{Ulam center} is a set $\Pi=\{\pi_1,\ldots,\pi_m\}$ of permutations of a common alphabet $\Sigma$. The problem asks to find a permutation $\sigma\in\mathcal P(\Sigma)$ that minimizes the radius $\rad_\Pi(\sigma):=\max_{i\in[m]} d_U(\sigma,\pi_i)$. We denote the optimum value of the instance $\Pi$ by
\[
    \OPT_{\mathrm{ctr}}(\Pi):=\min_{\sigma\in\mathcal P(\Sigma)}\rad_\Pi(\sigma).
\]

\paragraph{Approximation scheme.} For the Ulam median problem, an algorithm is a polynomial-time approximation scheme (PTAS) if, given any constant $\varepsilon > 0$ and an instance $\Pi$ of $m$ permutations over an alphabet of size $n$, it computes in polynomial time a solution $\rho$ satisfying $\cost_{\Pi}(\rho) \leq (1+\varepsilon)\OPT_{\mathrm{med}}(\Pi)$. An algorithm is an $\varepsilon mn$-additive approximation scheme for the Ulam median problem if it outputs a solution $\rho$ in polynomial time such that $\cost_{\Pi}(\rho) \leq \OPT_{\mathrm{med}}(\Pi) + \varepsilon mn$. The definition of PTAS for Ulam center is analogous to that of the median problem. An algorithm is an $\varepsilon n$-additive approximation scheme for the Ulam center problem if it outputs a solution $\tau$ in polynomial time such that $\rad_{\Pi}(\tau) \leq \OPT_{\mathrm{ctr}}(\Pi) + \varepsilon n$.

\paragraph{MAX-E3-LIN-2.} An instance of MAX-E3-LIN-2 consists of a variable set $\mathcal X=\{x_1,\ldots,x_k\}$ over $\mathbb F_2$ and a system $\mathcal E=\{E_1,\ldots,E_q\}$ of $q$ linear equations, each with exactly three variables. We write the $j$-th equation as $$E_j:\quad x_{j,1}\oplus x_{j,2}\oplus x_{j,3}=\lambda_j$$
where $\lambda_j\in\{0,1\}$ and $x_{j,1},x_{j,2},x_{j,3}\in\mathcal X$ denote the three \emph{variable
occurrences} in $E_j$. For an assignment $F:\mathcal X\to\mathbb F_2$, let $\val_{\mathrm{LIN}}(\mathcal E, F)$ denote the number of equations satisfied by $F$. The problem is to find an assignment that maximizes this value. We denote the optimum value of the instance $\mathcal E$ by
\[
    \OPT_{\mathrm{LIN}}(\mathcal E):=\max_{F:\mathcal X\to\mathbb F_2}\val_{\mathrm{LIN}}(\mathcal E, F).
\]
For $\mathbf u=(u_1,u_2,u_3)\in\{0,1\}^3$, let us define $\parity(\mathbf u):=u_1\oplus u_2\oplus u_3$. We will use the following hardness result of H{\aa}stad~\cite{Hastad01}.

\begin{theorem}[\cite{Hastad01}]
\label{thm:hastad_lin}
Unless $P = NP$, for any $\delta>0$, there is no polynomial-time algorithm that can distinguish whether for an instance $\mathcal E$ of MAX-E3-LIN-2, with $q$ equations,
\begin{itemize}
    \item $\OPT_{\mathrm{LIN}}(\mathcal E)\geq(1-\delta)q$, or
    \item $\OPT_{\mathrm{LIN}}(\mathcal E)\leq(1/2+\delta)q$.
\end{itemize}
\end{theorem}

\section{Hardness of Approximation for Ulam Median}
\label{sec:median-hardness}
In this section we show that the Ulam median problem is unlikely to have a PTAS; more specifically, we show the following.
\FourInputHardnessTheorem*

We prove Theorem~\ref{thm:four-input-hardness} by a gap-preserving reduction from MAX-E3-LIN-2. The reduction has three components. The \emph{equation gadget} assigns a one-unit advantage to subsets encoding an assignment that satisfies an equation. The \emph{variable-consistency gadget} rewards choosing the same truth value across all occurrences of a variable. Finally, a long common \emph{anchor block} separates these two gadgets within an optimal median. 

\subsection{A key component: a parity-sensitive local gadget}
\label{subsec:warmup}

Before presenting the full reduction, we isolate the elementary three-symbol gadget underlying the equation gadget, and explain how it can be used to distinguish whether a parity constraint is satisfied. Consider the set $\{a,b,c\}$ of three generic symbols and two permutations
\[
P_1:=abc \qquad \text{and} \qquad P_2:=bca
\]
of this set. For a subset $T\subseteq\{a,b,c\}$, define
\[
    h(T):=\max_{\tau\in\mathcal P(T)}\bigl(\LCS(\tau,P_1|_T)+\LCS(\tau,P_2|_T)\bigr).
\]
The individual values of $h(T)$ depend on more than the cardinality or parity of $T$. Nevertheless, when summed over an entire parity class, they exhibit the one-unit gap that we will need in the reduction.

\begin{lemma}\label{lem:score-table}
The sum of $h(T)$ over the four odd-cardinality subsets of $\{a,b,c\}$ is $11$, whereas its sum over the four even-cardinality subsets is $10$.
\end{lemma}

\begin{proof}
By Lemma~\ref{lem:two-permutation-score}, we have for any set $T\subseteq \{a,b,c\}$ that $h(T)=|T|+\LCS(P_1|_T,P_2|_T)$. Evaluating all eight subsets of $\{a,b,c\}$ gives
\[
\begin{array}{c|cccccccc}
T & \varnothing & \{a\} & \{b\} & \{c\} & \{a,b\} & \{a,c\} & \{b,c\} & \{a,b,c\} \\ \hline
h(T) & 0 & 2 & 2 & 2 & 3 & 3 & 4 & 5
\end{array}
\]
Hence the odd-cardinality values sum to $2+2+2+5=11$, while the even-cardinality values sum to $0+3+3+4=10$.
\end{proof}

Lemma~\ref{lem:score-table} suggests the following way to encode a three-variable parity equation. Consider a bit vector $\mathbf f=(f_1,f_2,f_3)\in\{0,1\}^3$ and a target parity $\lambda\in\{0,1\}$. We would like to associate with $\mathbf f$ four subsets of $\{a,b,c\}$ such that these subsets are precisely the four odd-cardinality subsets when $\parity(\mathbf f)=\lambda$ (correspondingly, the parity equation is satisfied), and precisely the four even-cardinality subsets otherwise (the equation is not satisfied). The total contribution of the four copies of the local gadget would then be $11$ in the former case and $10$ in the latter. To obtain this behavior, define the affine parity class
\[
    R_\lambda:=\{\mathbf r=(r_1,r_2,r_3)\in\{0,1\}^3:\parity(\mathbf r)=1\oplus\lambda\}.
\]
It is easy to check that $|R_{\lambda}| = 4$ for either value of $\lambda$. For each $\mathbf r\in R_\lambda$, associate with $\mathbf f$ the subset of $\{a,b,c\}$ whose indicator vector is $\mathbf f\oplus\mathbf r$. Since
\[
    \parity(\mathbf f\oplus\mathbf r)=\parity(\mathbf f)\oplus\parity(\mathbf r)=\parity(\mathbf f)\oplus(1\oplus\lambda),
\]
all four resulting subsets have odd cardinality if $\parity(\mathbf f)=\lambda$, and even cardinality otherwise. As $\mathbf r$ ranges over $R_\lambda$, the indicator vectors $\mathbf f\oplus\mathbf r$ range over the entire odd-parity class when $\parity(\mathbf f)=\lambda$, and over the entire even-parity class when $\parity(\mathbf f)\neq\lambda$. 

By a slight abuse of notation, let $h(\mathbf u)$ denote $h(T_{\mathbf u})$ for the subset $T_{\mathbf u}\subseteq\{a,b,c\}$ having indicator vector $\mathbf u$. Lemma~\ref{lem:score-table} gives
\[
    \sum_{\mathbf r\in R_\lambda} h(\mathbf f\oplus\mathbf r)
    =
    \begin{cases}
        11, & \text{if }\parity(\mathbf f)=\lambda,\\
        10, & \text{if }\parity(\mathbf f)\neq\lambda.
    \end{cases}
\]

This symmetrization is one of the reasons for using the entire parity class $R_\lambda$. A single subset does not have a score determined solely by its parity; for example, $h(\{a\})=2$ whereas $h(\{a,b,c\})=5$, although both sets have odd cardinality. By translating the entire parity class $R_\lambda$, the construction averages away this dependence on the particular local assignment and retains only whether the parity equation is satisfied. Thus, each equation will contribute a fixed baseline of $10$, together with one additional unit exactly when it is satisfied.

The same choice of $R_\lambda$ has a second property that will be crucial for enforcing consistency across different occurrences of a variable. For every coordinate $t\in\{1,2,3\}$, exactly two vectors in $R_\lambda$ have $r_t=0$ and exactly two have $r_t=1$. Hence, when the four copies corresponding to $R_\lambda$ are introduced for an equation, each variable occurrence appears equally often with polarity $0$ and polarity $1$. This balance is crucial for the variable-consistency gadget.

\subsection{The reduction}
\label{sec:reduction}
Recall that an instance of MAX-E3-LIN-2 consists of a variable set $\mathcal X=\{x_1,\ldots,x_k\}$ over $\mathbb F_2$ and a system $\mathcal E=\{E_1,\ldots,E_q\}$ of $q$ linear equations, each with exactly three variables, where $$E_j:\quad x_{j,1}\oplus x_{j,2}\oplus x_{j,3}=\lambda_j$$
for $\lambda_j\in\{0,1\}$ and $x_{j,1},x_{j,2},x_{j,3}\in\mathcal X$. From now on, we refer to the instance by $\mathcal E$.
\paragraph{Equation gadget.} For each equation $E_j:x_{j,1}\oplus x_{j,2}\oplus x_{j,3}=\lambda_j$, define
\[
    R_j:=\bigl\{\mathbf r=(r_1,r_2,r_3)\in\{0,1\}^3:r_1\oplus r_2\oplus r_3=1\oplus\lambda_j\bigr\}.
\]
For every $\mathbf r\in R_j$, introduce the three-symbol alphabet $\Sigma_{j,\mathbf r}=\{a_{j,\mathbf r},b_{j,\mathbf r},c_{j,\mathbf r}\}$, whose symbols correspond, respectively, to the three variable occurrences $x_{j,1},x_{j,2},x_{j,3}$ in $E_j$. Let $\Sigma_j:=\bigcup_{\mathbf r\in R_j}\Sigma_{j,\mathbf r}$ and $\Sigma:=\bigcup_{j\in[q]}\Sigma_j$. Since $|R_j| = 4$ for every $j\in [q]$, we have $|\Sigma_j|=12$ and $|\Sigma|=12q$.

Let $\var: \Sigma \to \mathcal{X}$ be the function mapping each symbol to its corresponding variable occurrence. Specifically, $\var(a_{j,\mathbf{r}}) = x_{j,1}$, $\var(b_{j,\mathbf{r}}) = x_{j,2}$, and $\var(c_{j,\mathbf{r}}) = x_{j,3}$. Define the polarity function $\pol:\Sigma\to\{0,1\}$ by $\pol(a_{j,\mathbf r})=r_1$, $\pol(b_{j,\mathbf r})=r_2$, and $\pol(c_{j,\mathbf r})=r_3$.

For each $\Sigma_{j,\mathbf r}$, consider the two local permutations
\[
    P_1(\Sigma_{j,\mathbf r}):=a_{j,\mathbf r}b_{j,\mathbf r}c_{j,\mathbf r},\qquad P_2(\Sigma_{j,\mathbf r}):=b_{j,\mathbf r}c_{j,\mathbf r}a_{j,\mathbf r}.
\]
Fix an arbitrary order $(\mathbf r^{(1)},\ldots,\mathbf r^{(4)})$ of vectors of $R_j$ and define
\begin{align*}
    A_{1j}&:=P_1(\Sigma_{j,\mathbf r^{(1)}})P_1(\Sigma_{j,\mathbf r^{(2)}})P_1(\Sigma_{j,\mathbf r^{(3)}})P_1(\Sigma_{j,\mathbf r^{(4)}}),\\
    A_{2j}&:=P_2(\Sigma_{j,\mathbf r^{(1)}})P_2(\Sigma_{j,\mathbf r^{(2)}})P_2(\Sigma_{j,\mathbf r^{(3)}})P_2(\Sigma_{j,\mathbf r^{(4)}}).
\end{align*}
The two equation-gadget permutations are
\[
    \alpha_1:=A_{11}A_{12}\cdots A_{1q},\qquad \alpha_2:=A_{21}A_{22}\cdots A_{2q}.
\]

\paragraph{Variable-consistency gadget.} For a variable $x\in\mathcal X$ and a polarity $p\in\{0,1\}$, let
\[
    \Sigma_x^p:=\{\xi\in\Sigma:\var(\xi)=x,\ \pol(\xi)=p\}.
\]
Fix an arbitrary total order of the symbols in $\Sigma$, and from now on, consider this fixed ordering. Let $B_x^p$ denote the permutation of $\Sigma_x^p$ induced by this order. With $\mathcal X=\{x_1,\ldots,x_n\}$, define variable blocks $B_{x_i}^0B_{x_i}^1$ and $B_{x_i}^1B_{x_i}^0$ for each variable $x_i$ and define
\[
    \beta_1:=B_{x_1}^0B_{x_1}^1\cdots B_{x_k}^0B_{x_k}^1,\qquad \beta_2:=B_{x_1}^1B_{x_1}^0\cdots B_{x_k}^1B_{x_k}^0.
\]

For each coordinate $t\in\{1,2,3\}$, exactly two vectors in $R_j$ satisfy $r_t=0$, and exactly two satisfy $r_t=1$. Hence every variable occurrence contributes exactly two symbols of each polarity. Consequently, for every $x\in\mathcal X$,
\begin{equation}\label{eq:polarity-balance}
    |\Sigma_x^0|=|\Sigma_x^1|=:r_x,\qquad\text{and}\qquad \sum_{x\in\mathcal X}r_x=\frac{|\Sigma|}{2}=6q.
\end{equation}

\paragraph{Anchor block and input permutations.} Let $N:=|\Sigma|=12q$. Consider a permutation $Z$ over an alphabet $\Sigma_Z$ that is disjoint from $\Sigma$ and of size $M:=|\Sigma_Z|=2N+1$. We refer to $Z$ as an \emph{anchor block} and the symbols in the alphabet $\Sigma_Z$ as \emph{anchor symbols}.

Construct four input permutations as follows:
\[
    \pi_1:=\alpha_1Z,\qquad \pi_2:=\alpha_2Z,\qquad \sigma_1:=Z\beta_1,\qquad \sigma_2:=Z\beta_2.
\]
Each permutation has length $N+M=36q+1$. We denote the resulting instance by $$\Pi_{\mathcal E}:=\{\pi_1,\pi_2,\sigma_1,\sigma_2\}.$$

\subsection{Gadget scores}

We first establish a block-decomposition principle that we will use repeatedly.
\begin{lemma}[Common block order]\label{lem:block-decomposition}
Let $V=V_1\sqcup\cdots\sqcup V_t$, and for $\ell\in\{1,2\}$ let $Q_\ell=Q_{\ell,1}Q_{\ell,2}\cdots Q_{\ell,t}$, where $Q_{\ell,s}\in\mathcal P(V_s)$. Then, for every $S\subseteq V$,
\[
    \max_{\tau\in\mathcal P(S)}\sum_{\ell=1}^2\LCS(\tau,Q_\ell|_S)=\sum_{s=1}^t\max_{\tau_s\in\mathcal P(S\cap V_s)}\sum_{\ell=1}^2\LCS(\tau_s,Q_{\ell,s}|_{S\cap V_s}).
\]
\end{lemma}

\begin{proof}
For each $s\in[t]$, let $S_s:=S\cap V_s$.  Since the sets
$V_1,\ldots,V_t$ are pairwise disjoint and each $Q_\ell$ lists them in the
same block order, restriction to $S$ gives
\[
 Q_\ell|_S
   =(Q_{\ell,1}|_{S_1})(Q_{\ell,2}|_{S_2})\cdots
     (Q_{\ell,t}|_{S_t})
 \qquad\text{for }\ell\in\{1,2\}.
\]
Let us first observe a blockwise property of LCS that we use later in the proof.  Suppose that
$A=A_1A_2\cdots A_t$ and $B=B_1B_2\cdots B_t$, where $A_s$ and
$B_s$ use only symbols from $V_s$.  For each $s$, choose a longest common
subsequence of $A_s$ and $B_s$.  Concatenating these $t$ subsequences in
block order produces a common subsequence of $A$ and $B$.  Hence
\[
 \LCS(A,B)\geq\sum_{s=1}^t\LCS(A_s,B_s).
\]
Conversely, let $C$ be any common subsequence of $A$ and $B$.  Because both
$A$ and $B$ visit the blocks in the order $V_1,\ldots,V_t$, the symbols of
$C$ also visit these blocks in nondecreasing order.  Thus $C$ can be written
as $C=C_1C_2\cdots C_t$, where $C_s$ consists of the symbols of $C$ lying in
$V_s$.  The sequence $C_s$ is a common subsequence of $A_s$ and $B_s$, so
\[
 |C|=\sum_{s=1}^t|C_s|
     \leq\sum_{s=1}^t\LCS(A_s,B_s).
\]
Maximizing over $C$ proves the exact identity
\begin{equation}\label{eq:blockwise-lcs}
 \LCS(A,B)=\sum_{s=1}^t\LCS(A_s,B_s).
\end{equation}

We now argue the lower-bound direction of the claimed equality.  For each
$s\in[t]$, choose a permutation $\tau_s^*\in\mathcal P(S_s)$ that maximizes
\[
 \sum_{\ell=1}^2
   \LCS(\tau_s,Q_{\ell,s}|_{S_s})
 \quad\text{over }\tau_s\in\mathcal P(S_s),
\]
and concatenate these local optimizers in block order:
$\tau^*:=\tau_1^*\tau_2^*\cdots\tau_t^*$.  Applying
\eqref{eq:blockwise-lcs} separately for $\ell=1$ and $\ell=2$ gives
\[
\begin{aligned}
 \sum_{\ell=1}^2\LCS(\tau^*,Q_\ell|_S)
 &=\sum_{\ell=1}^2\sum_{s=1}^t
      \LCS(\tau_s^*,Q_{\ell,s}|_{S_s})\\
 &=\sum_{s=1}^t\max_{\tau_s\in\mathcal P(S_s)}
      \sum_{\ell=1}^2\LCS(\tau_s,Q_{\ell,s}|_{S_s}).
\end{aligned}
\]
Since $\tau^*\in\mathcal P(S)$, we derive that
\[
    \max_{\tau\in\mathcal P(S)}\sum_{\ell=1}^2\LCS(\tau,Q_\ell|_S) \ge \sum_{s=1}^t\max_{\tau_s\in\mathcal P(S_s)}\sum_{\ell=1}^2\LCS(\tau_s,Q_{\ell,s}|_{S_s}).
\]

For the reverse direction, fix an arbitrary $\tau\in\mathcal P(S)$.  Obtain
$\overline\tau$ by stably sorting the symbols of $\tau$ according to the
block order $V_1,\ldots,V_t$: symbols from $V_s$ precede symbols from
$V_{s+1}$, while the relative order of symbols belonging to the same block
is unchanged.  We claim that, for each $\ell\in\{1,2\}$,
\begin{equation}\label{eq:stable-sort-preserves-lcs}
 \LCS(\tau,Q_\ell|_S)
 \leq\LCS(\overline\tau,Q_\ell|_S).
\end{equation}
Indeed, take a longest common subsequence $C_\ell$ of $\tau$ and
$Q_\ell|_S$.  Since $Q_\ell|_S$ visits the blocks in the order
$V_1,\ldots,V_t$, the block indices of the symbols along $C_\ell$ are
nondecreasing.  Stable sorting places symbols from different blocks in
exactly this order and preserves their relative order within every block.
Consequently, $C_\ell$ is a subsequence of
$\overline\tau$, establishing
\eqref{eq:stable-sort-preserves-lcs}.

For $s\in[t]$, let $\overline\tau_s:=\tau|_{S_s}$, where the restriction
retains the order inherited from $\tau$.  By the definition of the stable
sort,
$\overline\tau=\overline\tau_1\overline\tau_2\cdots\overline\tau_t$.
So we get
\[
\begin{aligned}
 \sum_{\ell=1}^2\LCS(\tau,Q_\ell|_S)
 &\leq\sum_{\ell=1}^2\LCS(\overline\tau,Q_\ell|_S) &&\text{(by~\eqref{eq:stable-sort-preserves-lcs})}\\
 &=\sum_{s=1}^t\sum_{\ell=1}^2
      \LCS(\overline\tau_s,Q_{\ell,s}|_{S_s})&&\text{(by~\eqref{eq:blockwise-lcs})}\\
 &\leq\sum_{s=1}^t\max_{\tau_s\in\mathcal P(S_s)}
      \sum_{\ell=1}^2\LCS(\tau_s,Q_{\ell,s}|_{S_s}).
\end{aligned}
\]
This bound holds for every $\tau\in\mathcal P(S)$, and therefore also after
maximizing over $\tau$.  Together with the lower-bound direction, it proves the
lemma.  Note that if some $S_s$ is empty, the same argument applies with the unique
empty permutation and a zero LCS contribution for that block.
\end{proof}

For $S\subseteq\Sigma$, define
\[
    u(S):=\sum_{j\in[q]}\max_{\tau_j\in\mathcal P(S\cap\Sigma_j)}\bigl(\LCS(\tau_j,A_{1j}|_S)+\LCS(\tau_j,A_{2j}|_S)\bigr),
\]
and
\[
    v(S):=\sum_{x\in\mathcal X}\max\{|S\cap\Sigma_x^0|,|S\cap\Sigma_x^1|\}.
\]

\begin{lemma}\label{lem:equation-score}
For every $S\subseteq\Sigma$,
\begin{equation}
    \label{eq:score}
    \max_{\tau\in\mathcal P(S)}\bigl(\LCS(\tau,\alpha_1|_S)+\LCS(\tau,\alpha_2|_S)\bigr)=u(S).
\end{equation}

Moreover, if $F:\mathcal X\to\mathbb F_2$ is an assignment and $S^F:=\{\xi\in\Sigma:F(\var(\xi))\oplus\pol(\xi)=1\}$ then
\[
    u(S^F)=10q+\val_{\mathrm{LIN}}(\mathcal E, F).
\]
\end{lemma}

\begin{proof}
To establish the identity~\eqref{eq:score}, apply Lemma~\ref{lem:block-decomposition} with
$V_j=\Sigma_j$ and $Q_{\ell,j}=A_{\ell j}$ for $\ell\in\{1,2\}$.
Indeed, the two permutations
$\alpha_\ell=A_{\ell 1}A_{\ell 2}\cdots A_{\ell q}$, for $\ell \in \{1,2\}$, list the disjoint
equation blocks $\Sigma_1,\ldots,\Sigma_q$ in the same order.  The lemma
therefore gives
\[
\begin{aligned}
 &\max_{\tau\in\mathcal P(S)}
   \bigl(\LCS(\tau,\alpha_1|_S)+\LCS(\tau,\alpha_2|_S)\bigr)\\
 &\qquad={}
   \sum_{j\in[q]}\max_{\tau_j\in\mathcal P(S\cap\Sigma_j)}
   \bigl(\LCS(\tau_j,A_{1j}|_S)+\LCS(\tau_j,A_{2j}|_S)\bigr),
\end{aligned}
\]
and the right-hand side is exactly $u(S)$.

We next refine the decomposition of $u(S)$ into the four copies belonging
to each equation; afterward we specialize to $S=S^F$. Consider the set $\{a,b,c\}$ of generic symbols introduced in Subsection~\ref{subsec:warmup}. For each $j\in[q]$
and $\mathbf r\in R_j$, let us consider the following bijection
\[
 \iota_{j,\mathbf r}:\Sigma_{j,\mathbf r}\longrightarrow\{a,b,c\}
\]
that maps $a_{j,\mathbf r},b_{j,\mathbf r},c_{j,\mathbf r}$
to $a,b,c$, respectively. By abusing notation, extend $\iota_{j,\mathbf r}$ symbolwise to
sequences over $\Sigma_{j,\mathbf r}$.  Fix $j$, and abbreviate
\[
 S_j:=S\cap\Sigma_j,
 \qquad
 S_{j,\mathbf r}:=S\cap\Sigma_{j,\mathbf r},
 \qquad
 T_{j,\mathbf r}:=\iota_{j,\mathbf r}(S_{j,\mathbf r})
 \subseteq\{a,b,c\}.
\]
Since $A_{\ell j}$ is a sequence of alphabet $\Sigma_j$, $A_{\ell j}|_S$ is the same as $A_{\ell j}|_{S_j}$.  Recall that the four vectors of $R_j$ were placed in
a fixed order
$(\mathbf r^{(1)},\ldots,\mathbf r^{(4)})$.  Restricting $A_{\ell j}$ to
$S_j$ preserves this common block order and gives, for $\ell\in\{1,2\}$,
\[
 A_{\ell j}|_{S_j}
 =P_\ell(\Sigma_{j,\mathbf r^{(1)}})|_{S_{j,\mathbf r^{(1)}}}
  \cdots
  P_\ell(\Sigma_{j,\mathbf r^{(4)}})|_{S_{j,\mathbf r^{(4)}}}.
\]
Lemma~\ref{lem:block-decomposition}, applied to these four disjoint
subblocks with ambient alphabet $\Sigma_j$ and subset $S_j$, therefore gives
\begin{equation}
\label{eq:equation-score-fixed-j-decomposition}
\begin{aligned}
 &\max_{\tau_j\in\mathcal P(S_j)}
   \sum_{\ell=1}^2\LCS(\tau_j,A_{\ell j}|_{S_j})=
   \sum_{\mathbf r\in R_j}
   \max_{\tau_{j,\mathbf r}\in\mathcal P(S_{j,\mathbf r})}
   \sum_{\ell=1}^2
   \LCS\bigl(\tau_{j,\mathbf r},
      P_\ell(\Sigma_{j,\mathbf r})|_{S_{j,\mathbf r}}\bigr).
\end{aligned}
\end{equation}

It remains to identify each summand on the right. Recall that $P_1=abc$ and $P_2=bca$. Observe that for every locally optimal permutation $\tau_{j,\mathbf r}$ and $\ell\in\{1,2\}$, 
\[
 \iota_{j,\mathbf r}\bigl(
   P_\ell(\Sigma_{j,\mathbf r})|_{S_{j,\mathbf r}}\bigr)
 =P_\ell|_{T_{j,\mathbf r}}.
\]
Moreover, a sequence $C$ is a common subsequence of
$\tau_{j,\mathbf r}$ and
$P_\ell(\Sigma_{j,\mathbf r})|_{S_{j,\mathbf r}}$ if and only if
$\iota_{j,\mathbf r}(C)$ is a common subsequence of
$\iota_{j,\mathbf r}(\tau_{j,\mathbf r})$ and
$P_\ell|_{T_{j,\mathbf r}}$.  This correspondence preserves sequence
lengths and is reversible using $\iota_{j,\mathbf r}^{-1}$; hence
\[
 \LCS\bigl(\tau_{j,\mathbf r},
    P_\ell(\Sigma_{j,\mathbf r})|_{S_{j,\mathbf r}}\bigr)
 =\LCS\bigl(\iota_{j,\mathbf r}(\tau_{j,\mathbf r}),
    P_\ell|_{T_{j,\mathbf r}}\bigr).
\]
Finally, $\tau_{j,\mathbf r}\mapsto
\iota_{j,\mathbf r}(\tau_{j,\mathbf r})$ is a bijection from the set of permutations
$\mathcal P(S_{j,\mathbf r})$ to $\mathcal P(T_{j,\mathbf r})$.
Consequently, the local maximum in
\eqref{eq:equation-score-fixed-j-decomposition} is
\[
 \max_{\sigma\in\mathcal P(T_{j,\mathbf r})}
 \bigl(\LCS(\sigma,P_1|_{T_{j,\mathbf r}})
      +\LCS(\sigma,P_2|_{T_{j,\mathbf r}})\bigr)
 =h(T_{j,\mathbf r}).
\]
Substituting this identity into
\eqref{eq:equation-score-fixed-j-decomposition} gives, for each fixed $j$,
\[
 \max_{\tau_j\in\mathcal P(S_j)}
 \sum_{\ell=1}^2\LCS(\tau_j,A_{\ell j}|_{S_j})
 =\sum_{\mathbf r\in R_j}
 h\bigl(\iota_{j,\mathbf r}(S\cap\Sigma_{j,\mathbf r})\bigr).
\]
Finally, summing this equality over $j\in[q]$ and using the definition of
$u(S)$ yields
\begin{equation}
\label{eq:equation-score-copy-decomposition}
 u(S)=\sum_{j\in[q]}\sum_{\mathbf r\in R_j}
 h\bigl(\iota_{j,\mathbf r}(S\cap\Sigma_{j,\mathbf r})\bigr).
\end{equation}

Next, fix an equation
$E_j:x_{j,1}\oplus x_{j,2}\oplus x_{j,3}=\lambda_j$, and let
\[
 \mathbf f_j:=\bigl(F(x_{j,1}),F(x_{j,2}),F(x_{j,3})\bigr).
\]
In $A_{\ell j}$, the block indexed by $\mathbf r=(r_1,r_2,r_3)$, the symbol corresponding
to $x_{j,t}$ belongs to $S^F$ precisely when
$F(x_{j,t})\oplus r_t=1$.  Thus the indicator vector of
$\iota_{j,\mathbf r}(S^F\cap\Sigma_{j,\mathbf r})$ is
$\mathbf f_j\oplus\mathbf r$.  Moreover, every $\mathbf r\in R_j$
satisfies $\parity(\mathbf r)=1\oplus\lambda_j$, so
\begin{equation}
\label{eq:equation-score-parity-translation}
 \parity(\mathbf f_j\oplus\mathbf r)
 =\parity(\mathbf f_j)\oplus(1\oplus\lambda_j).
\end{equation}
The translation $\mathbf r\mapsto\mathbf f_j\oplus\mathbf r$ is injective,
and both $R_j$ and each parity class in $\{0,1\}^3$ have four elements.
Consequently, its four images are the
entire parity class specified by
\eqref{eq:equation-score-parity-translation}; as illustrated
in Figure~\ref{fig:equation-score}.

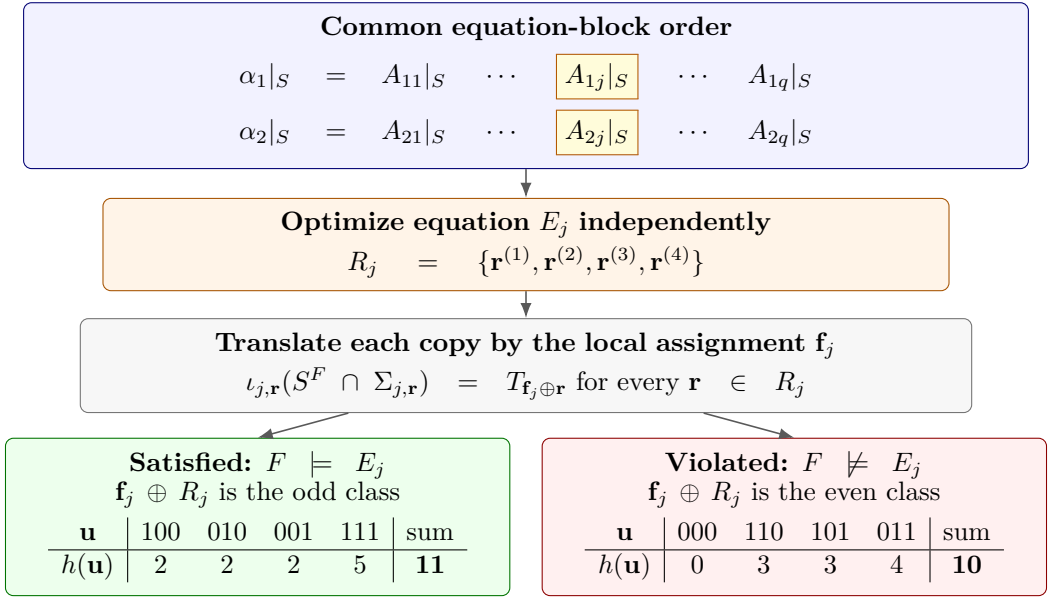
\begin{figure}[H]
\centering
\begingroup
\setlength{\fboxsep}{3pt}
\begin{tikzpicture}[
  font=\small,>=Latex,
  flow/.style={-{Latex[length=2mm]},semithick,draw=black!65},
  globalbox/.style={draw=blue!45!black,rounded corners=3pt,
    fill=blue!5,align=center,text width=12.8cm,
    inner xsep=7pt,inner ysep=5pt},
  copybox/.style={draw=orange!65!black,rounded corners=3pt,
    fill=orange!9,align=center,text width=10.7cm,
    inner xsep=7pt,inner ysep=5pt},
  mapbox/.style={draw=black!55,rounded corners=3pt,
    fill=black!3,align=center,text width=11.3cm,
    inner xsep=7pt,inner ysep=5pt},
  goodbox/.style={draw=green!45!black,rounded corners=3pt,
    fill=green!7,align=center,text width=6.25cm,
    inner xsep=6pt,inner ysep=5pt},
  badbox/.style={draw=red!55!black,rounded corners=3pt,
    fill=red!6,align=center,text width=6.25cm,
    inner xsep=6pt,inner ysep=5pt}
]
  \node[globalbox] (global) at (0,3.55)
    {\textbf{Common equation-block order}\\[1.5mm]
     $\alpha_1|_S=A_{11}|_S\;\cdots\;
       \fcolorbox{orange!70!black}{yellow!18}{$A_{1j}|_S$}\;
       \cdots\;A_{1q}|_S$\\[1mm]
     $\alpha_2|_S=A_{21}|_S\;\cdots\;
       \fcolorbox{orange!70!black}{yellow!18}{$A_{2j}|_S$}\;
       \cdots\;A_{2q}|_S$};

  \node[copybox] (copies) at (0,1.45)
    {\textbf{Optimize equation $E_j$ independently}\\[1mm]
     $R_j=\{\mathbf r^{(1)},\mathbf r^{(2)},
                \mathbf r^{(3)},\mathbf r^{(4)}\}$};

  \node[mapbox] (map) at (0,-.15)
    {\textbf{Translate each copy by the local assignment $\mathbf f_j$}\\[1mm]
     $\iota_{j,\mathbf r}(S^F\cap\Sigma_{j,\mathbf r})
       =T_{\mathbf f_j\oplus\mathbf r}$
     for every $\mathbf r\in R_j$};

  \node[goodbox] (good) at (-3.55,-2.15)
    {\textbf{Satisfied: $F\models E_j$}\\[-.5mm]
     $\mathbf f_j\oplus R_j$ is the odd class\\[1mm]
     $\begin{array}{c|cccc|c}
       \mathbf u&100&010&001&111&\text{sum}\\ \hline
       h(\mathbf u)&2&2&2&5&\mathbf{11}
     \end{array}$};

  \node[badbox] (bad) at (3.55,-2.15)
    {\textbf{Violated: $F\not\models E_j$}\\[-.5mm]
     $\mathbf f_j\oplus R_j$ is the even class\\[1mm]
     $\begin{array}{c|cccc|c}
       \mathbf u&000&110&101&011&\text{sum}\\ \hline
       h(\mathbf u)&0&3&3&4&\mathbf{10}
     \end{array}$};

  \draw[flow] (global) -- (copies);
  \draw[flow] (copies) -- (map);
  \draw[flow] ($(map.south)+(-2.35,0)$) -- (good.north);
  \draw[flow] ($(map.south)+(2.35,0)$) -- (bad.north);
\end{tikzpicture}
\endgroup
\caption{The two-level decomposition used in
Lemma~\ref{lem:equation-score}.  The aligned input rows decompose by
equation.  For a fixed equation $E_j$, its four subblocks are indexed by
$R_j$, and translation by $\mathbf f_j$ maps them bijectively onto an entire
parity class, whose local scores sum to $11$ or $10$.}
\label{fig:equation-score}
\end{figure}

If $F$ satisfies $E_j$, then $\parity(\mathbf f_j)=\lambda_j$, and
thus $\parity(\mathbf f_j)\oplus(1\oplus\lambda_j) = 1$.  The four sets
appearing in the inner sum of
\eqref{eq:equation-score-copy-decomposition} are therefore exactly the four
odd-cardinality subsets of $\{a,b,c\}$, whose $h$-values sum to $11$ by
Lemma~\ref{lem:score-table}.  If $F$ violates $E_j$, then
$\parity(\mathbf f_j)=1\oplus\lambda_j$, and
thus $\parity(\mathbf f_j)\oplus(1\oplus\lambda_j) = 0$. Then the four sets are
exactly the four even-cardinality subsets, whose $h$-values sum to $10$ by
Lemma~\ref{lem:score-table}.

Finally, since exactly $\val_{\mathrm{LIN}}(\mathcal E, F)$ of the $q$ equations are
satisfied,
\eqref{eq:equation-score-copy-decomposition} gives
\[
 u(S^F)
 =11\val_{\mathrm{LIN}}(\mathcal E, F)
  +10\bigl(q-\val_{\mathrm{LIN}}(\mathcal E, F)\bigr)
 =10q+\val_{\mathrm{LIN}}(\mathcal E, F).
\]
\end{proof}

\begin{lemma}\label{lem:beta_lcs}
For every $S\subseteq\Sigma$,
\[
    \max_{\tau\in\mathcal P(S)}\bigl(\LCS(\tau,\beta_1|_S)+\LCS(\tau,\beta_2|_S)\bigr)=|S|+v(S).
\]
\end{lemma}

\begin{proof}
Apply Lemma~\ref{lem:block-decomposition} to the common variable-block order $x_1,\ldots,x_k$. Fix a variable $x$, and abbreviate $S_x^p:=S\cap\Sigma_x^p$. Within the $x$-block, the two restricted permutations are
\[
    B_x^0|_S\,B_x^1|_S\qquad\text{and}\qquad B_x^1|_S\,B_x^0|_S.
\]
Since a common subsequence of these permutations cannot contain symbols from both polarity classes, their longest common subsequence has length $\max\{|S_x^0|,|S_x^1|\}$. Lemma~\ref{lem:two-permutation-score} therefore gives the local optimum
\[
    |S_x^0|+|S_x^1|+\max\{|S_x^0|,|S_x^1|\}.
\]
Summing over all variables proves the claim.
\end{proof}

\subsection{The anchor block}

Recall that $Z \in \mathcal P(\Sigma_Z)$ is the anchor block and the symbols in $\Sigma_Z$ are anchor symbols. Also, $\Sigma_Z$ is of size $M = 2N + 1$, where $N = |\Sigma|$ (the number of non-anchor symbols). Next, we argue that in our reduction, the anchor block ensures that the equation gadget and the variable-consistency gadget interact only through a partition of $\Sigma$.

\begin{lemma}\label{lem:anchor_block}
\[
    \max_{\tau\in\mathcal P(\Sigma\sqcup\Sigma_Z)}\sum_{\pi\in \Pi_{\mathcal E}}\LCS(\tau,\pi)=4M+\max_{(S_L,S_R):\Sigma=S_L\sqcup S_R}\bigl(u(S_L)+|S_R|+v(S_R)\bigr).
\]
\end{lemma}

\begin{proof}
We first prove the lower-bound direction. Fix a partition $(S_L, S_R)$ such that $\Sigma=S_L\sqcup S_R$, and choose $\tau_L\in\mathcal P(S_L)$ and $\tau_R\in\mathcal P(S_R)$ attaining the optima in Lemmas~\ref{lem:equation-score} and~\ref{lem:beta_lcs}, respectively. Consider a candidate permutation $\tau_LZ\tau_R$. 

Fix $i\in\{1,2\}$. Any common subsequence of $\tau_LZ\tau_R$ and $\alpha_iZ$ that contains an anchor symbol cannot contain a symbol of $\tau_R$, because every symbol of $\tau_R$ occurs after all anchors in the first permutation, whereas every non-anchor symbol occurs before all anchors in $\alpha_iZ$. Hence such a common subsequence has length at most $M+\LCS(\tau_L,\alpha_i|_{S_L})$,  which is attained by taking an LCS of $\tau_L$ and $\alpha_i|_{S_L}$ followed by all of $Z$. A common subsequence containing no anchor symbol has length at most $N<M$. Therefore,
\[
    \LCS(\tau_LZ\tau_R,\alpha_iZ)=M+\LCS(\tau_L,\alpha_i|_{S_L}).
\]
By a symmetric argument,
\[
    \LCS(\tau_LZ\tau_R,Z\beta_i)=M+\LCS(\tau_R,\beta_i|_{S_R}).
\]
Summing over the four permutations of our reduced instance $\Pi_{\mathcal E}$ and applying Lemmas~\ref{lem:equation-score} and~\ref{lem:beta_lcs} gives 
\[
\sum_{\pi\in \Pi_{\mathcal E}}\LCS(\tau_LZ\tau_R,\pi) = 4M+u(S_L)+|S_R|+v(S_R).
\]
Since the argument holds for any partition $(S_L,S_R)$, maximizing over the partition proves the lower bound
\[
    \max_{\tau\in\mathcal P(\Sigma\sqcup\Sigma_Z)}\sum_{\pi\in \Pi_{\mathcal E}}\LCS(\tau,\pi) \ge 4M+\max_{(S_L,S_R):\Sigma=S_L\sqcup S_R}\bigl(u(S_L)+|S_R|+v(S_R)\bigr).
\]

For the other direction, fix an arbitrary $\tau\in\mathcal P(\Sigma\sqcup\Sigma_Z)$. We first argue that it suffices to assume that the anchor symbols occur in $\tau$ in the order $Z$. Write $Z=z_1z_2\cdots z_M$, and let $\overline\tau$ be obtained from $\tau$ by leaving all non-anchor symbols $\Sigma$ in their original positions (as in $\tau$) while replacing the anchor symbols, from left to right, by $z_1,\ldots,z_M$. 

To see that this reordering of anchor symbols cannot decrease any of the four LCS values, consider an arbitrary permutation $\gamma$ of $\Sigma$. For any position $c$, let $\tau_{\leq c}$ and $\tau_{> c}$ denote the prefix and suffix of $\tau$ separated at the cut position $c$, respectively. Observe that if $c$ ranges over all cuts of $\tau$, including the two end cuts, then
\[
    \LCS(\tau,\gamma Z)=\max_c\left(\LCS(\tau_{\leq c}|_\Sigma,\gamma)+\LCS(\tau_{>c}|_{\Sigma_Z},Z)\right),
\]
and
\[
    \LCS(\tau,Z\gamma)=\max_c\left(\LCS(\tau_{\leq c}|_{\Sigma_Z},Z)+\LCS(\tau_{>c}|_\Sigma,\gamma)\right).
\]
For any cut position $c$ such that the prefix $\tau_{\leq c}$ contains $n_p$ anchor symbols, we have $\overline\tau_{>c}|_{\Sigma_Z}=z_{n_p+1}\cdots z_M$ and $\overline\tau_{\leq c}|_{\Sigma_Z}=z_1\cdots z_{n_p}$. Thus, the corresponding contributions of the anchor symbols in the LCS with $Z$ are exactly $M-n_p$ and $n_p$, respectively, which are also the largest possible values for a suffix and prefix containing those numbers of anchor symbols. Since the contribution of the non-anchor symbols in the LCS is unchanged, we have $\LCS(\overline\tau,\pi)\geq\LCS(\tau,\pi)$ for every $\pi\in \Pi_{\mathcal E}$.

It remains to bound the score of $\overline\tau$. Let $S_L$ be the set of symbols of $\Sigma$ occurring before the central anchor $z_{N+1}$ in $\overline\tau$, and let $S_R:=\Sigma\setminus S_L$. For a cut position $c$, let $n_p(c)$ denote the number of anchor symbols in the prefix $\overline \tau_{\leq c}$. For $i\in\{1,2\}$, we have
\[
    \LCS(\overline\tau,\alpha_iZ)=\max_c\left(\LCS(\overline\tau_{\leq c}|_\Sigma,\alpha_i)+M-n_p(c)\right).
\]
Any cut with $n_p(c)\geq N+1$ gives a value at most $N+M-(N+1)=M-1$, whereas the cut immediately before the first anchor gives a value at least $M$. Hence an optimizing cut satisfies $n_p(c)\leq N$. Its $\Sigma$-prefix is therefore a subsequence of $\overline\tau|_{S_L}$, so
\[
    \LCS(\overline\tau,\alpha_iZ)\leq M+\LCS(\overline\tau|_{S_L},\alpha_i|_{S_L}).
\]
Recall, for $i \in \{1,2\}$, $\pi_i = \alpha_i Z \in \Pi_{\mathcal E}$. Summing over $i$ and applying Lemma~\ref{lem:equation-score} gives
\begin{equation}
    \label{eq:upper-bound-anchor-block}
    \LCS(\overline\tau,\pi_1)+\LCS(\overline\tau,\pi_2)\leq 2M+u(S_L).
\end{equation}

Similarly,
\[
    \LCS(\overline\tau,Z\beta_i)=\max_c\left(n_p(c)+\LCS(\overline\tau_{>c}|_\Sigma,\beta_i)\right).
\] 
Any cut with $n_p(c)\leq N$ gives a value at most $N+N=2N=M-1$, whereas the cut immediately after the last anchor gives a value at least $M$. Hence an optimizing cut satisfies $n_p(c)\geq N+1$. Its $\Sigma$-suffix is therefore a subsequence of $\overline\tau|_{S_R}$, and consequently
\[
    \LCS(\overline\tau,Z\beta_i)\leq M+\LCS(\overline\tau|_{S_R},\beta_i|_{S_R}).
\]
Recall, for $i \in \{1,2\}$, $\sigma_i = Z \beta_i \in \Pi_{\mathcal E}$. Summing over $i$ and applying Lemma~\ref{lem:beta_lcs} gives
\begin{equation}
    \label{eq:lower-bound-anchor-block}
    \LCS(\overline\tau,\sigma_1)+\LCS(\overline\tau,\sigma_2)\leq 2M+|S_R|+v(S_R).
\end{equation}

Combining~\eqref{eq:upper-bound-anchor-block} and~\eqref{eq:lower-bound-anchor-block} yields
\[
    \sum_{\pi\in \Pi_{\mathcal E}}\LCS(\tau,\pi)\leq\sum_{\pi\in \Pi_{\mathcal E}}\LCS(\overline\tau,\pi)\leq 4M+u(S_L)+|S_R|+v(S_R).
\]
Maximizing over all $\tau\in\mathcal P(\Sigma\sqcup\Sigma_Z)$ shows the upper bound
\[
    \max_{\tau\in\mathcal P(\Sigma\sqcup\Sigma_Z)}\sum_{\pi\in \Pi_{\mathcal E}}\LCS(\tau,\pi) \le 4M+\max_{(S_L,S_R):\Sigma=S_L\sqcup S_R}\bigl(u(S_L)+|S_R|+v(S_R)\bigr)
\]
and that completes the proof of the lemma.
\end{proof}

\subsection{From partitions to assignments}

We will later use the following Lipschitz-type property of the function $h(\cdot)$ defined in Subsection~\ref{subsec:warmup}.
\begin{lemma}\label{lem:replacement}
For all $S,T\subseteq\{a,b,c\}$,
\[
    h(T)-h(S)\leq 2|T\setminus S|-|S\setminus T|.
\]
\end{lemma}

\begin{proof}
By Lemma~\ref{lem:two-permutation-score}, for any subset $U \subseteq \{a,b,c\}$,
\[
h(U)=|U|+\LCS(P_1|_U,P_2|_U).
\]
Adding one symbol to $U$ increases the cardinality term by one, while the LCS term cannot decrease and can increase by at most one. Hence, for each $x\in \{a,b,c\}\setminus U$,
\begin{equation}
    \label{eq:h-change}
    1\leq h(U\cup\{x\}) - h(U) \leq 2.
\end{equation}

\begin{align*}
    h(T) - h(S) & = h(T) - h(S\cap T) + h(S\cap T) - h(S)\\
    &= \left( h(T) - h(S\cap T)\right) - \left( h(S) - h(S\cap T)\right)\\
    &\le  2 |T\setminus S| - |S \setminus T|
\end{align*}
where the last inequality is derived by repeatedly adding symbols in $S \cap T$ one by one to obtain $T$ and applying the upper bound inequality of~\eqref{eq:h-change}, and by repeatedly adding symbols in $S \cap T$ one by one to obtain $S$ and applying the lower bound inequality of~\eqref{eq:h-change}.
\end{proof}

\begin{lemma}\label{lem:partition-score}
\[
    \max_{(S_L,S_R) : \Sigma=S_L\sqcup S_R}\bigl(|S_R|+u(S_L)+v(S_R)\bigr)=22q+\OPT_{\mathrm{LIN}}(\mathcal E).
\]
\end{lemma}

\begin{proof}
We first prove the lower-bound direction. Let $F:\mathcal X\to\mathbb F_2$ be an assignment satisfying exactly $\OPT_{\mathrm{LIN}}(\mathcal E)$ equations, and consider the following partition
\[
    S_L^F:=\{\xi\in\Sigma:F(\var(\xi))\oplus\pol(\xi)=1\},\qquad S_R^F:=\Sigma\setminus S_L^F.
\]
Recall that for a variable $x$ and polarity $p \in \{0,1\}$, $\Sigma_x^p=\{\xi\in\Sigma:\var(\xi)=x,\ \pol(\xi)=p\}$. Then observe that for $F(x)=p$, the symbols in $\Sigma_x^{1-p}$ lie in $S_L^F$, while those in $\Sigma_x^p$ lie in $S_R^F$. By~\eqref{eq:polarity-balance}, each of these sets has size $r_x$. Consequently, $|S_L^F|=|S_R^F|=6q$ and $v(S_R^F)=\sum_x r_x=6q$. Lemma~\ref{lem:equation-score} gives $u(S_L^F)=10q+\OPT_{\mathrm{LIN}}(\mathcal E)$. Therefore,
\[
    |S_R^F|+u(S_L^F)+v(S_R^F)=22q+\OPT_{\mathrm{LIN}}(\mathcal E).
\]
Thus clearly,
\[
    \max_{(S_L,S_R) : \Sigma=S_L\sqcup S_R}\bigl(|S_R|+u(S_L)+v(S_R)\bigr) \ge 22q+\OPT_{\mathrm{LIN}}(\mathcal E).
\]

For the other direction, fix an arbitrary partition $(S_L, S_R)$ such that $\Sigma=S_L\sqcup S_R$. For each $x\in\mathcal X$, choose a polarity $p_x\in\{0,1\}$ such that
\[
    |S_R\cap\Sigma_x^{p_x}|=\max\{|S_R\cap\Sigma_x^0|,|S_R\cap\Sigma_x^1|\},
\]
breaking ties arbitrarily, and then consider an assignment $F:\mathcal X\to\mathbb F_2$ that sets $F(x)=p_x$  for every $x \in \mathcal X$. Call a symbol $\xi$ \emph{true} if $F(\var(\xi))\oplus\pol(\xi)=1$, and \emph{false} otherwise. Let $S^F$ denote the set of all true symbols. Let $D$ be the number of false symbols in $S_L$, and let $U$ be the number of true symbols in $S_R$. For each variable $x$, its false symbols are precisely $\Sigma_x^{p_x}$, and by the definition of $p_x$, the number of false symbols in $S_R$ is exactly the contribution of variable $x$ to $v(S_R)$. Since there are $\sum_x r_x=6q$ false symbols in total, we have $D=6q-v(S_R)$ and $U=|S_R|-v(S_R)$. 

As in the proof of Lemma~\ref{lem:equation-score}, consider the set $\{a,b,c\}$ of generic symbols introduced in Subsection~\ref{subsec:warmup}. For each $j\in[q]$
and $\mathbf r\in R_j$, let us consider the following bijection
\[
 \iota_{j,\mathbf r}:\Sigma_{j,\mathbf r}\longrightarrow\{a,b,c\}
\]
that maps $a_{j,\mathbf r},b_{j,\mathbf r},c_{j,\mathbf r}$
to $a,b,c$, respectively. By abusing notation, extend $\iota_{j,\mathbf r}$ symbolwise to
sequences over $\Sigma_{j,\mathbf r}$.
For every $j\in[q]$ and $\mathbf r\in R_j$, define
\[
    C_{j,\mathbf r}:=\iota_{j,\mathbf r}(S^F\cap\Sigma_{j,\mathbf r}),\qquad T_{j,\mathbf r}:=\iota_{j,\mathbf r}(S_L\cap\Sigma_{j,\mathbf r}).
\]
Applying Lemma~\ref{lem:replacement} with $S=C_{j,\mathbf r}$ and $T=T_{j,\mathbf r}$ gives
\[
    h(T_{j,\mathbf r})\leq h(C_{j,\mathbf r})+2|T_{j,\mathbf r}\setminus C_{j,\mathbf r}|-|C_{j,\mathbf r}\setminus T_{j,\mathbf r}|.
\]
Summing over all $j$ and $\mathbf r$, the first set-difference term counts exactly the $D$ false symbols placed in $S_L$; more specifically,
\[
\sum_{j, \mathbf r} |T_{j,\mathbf r}\setminus C_{j,\mathbf r}| = D,
\]
while the second set-difference term counts exactly the $U$ true symbols placed in $S_R$, i.e.,
\[
\sum_{j, \mathbf r} |C_{j,\mathbf r}\setminus T_{j,\mathbf r}| = U.
\]
Together with the definition of $h(\cdot)$ and Lemma~\ref{lem:equation-score}, this yields
\[
\begin{aligned}
    u(S_L)&\leq u(S^F)+2D-U\\
    &=10q+\val_{\mathrm{LIN}}(\mathcal E, F)+2(6q-v(S_R))-(|S_R|-v(S_R)).
\end{aligned}
\]
Rearranging, we obtain
\[
    |S_R|+u(S_L)+v(S_R)\leq22q+\val_{\mathrm{LIN}}(\mathcal E, F)\leq22q+\OPT_{\mathrm{LIN}}(\mathcal E).
\]
Since the partition $(S_L, S_R)$ we consider is arbitrary, it follows that
\[
    \max_{(S_L,S_R) : \Sigma=S_L\sqcup S_R}\bigl(|S_R|+u(S_L)+v(S_R)\bigr) \le 22q+\OPT_{\mathrm{LIN}}(\mathcal E)
\]
and that concludes the proof.
\end{proof}

\subsection{Optimum value and hardness proofs}

We now relate the optimum value of the constructed Ulam median instance to that of the original MAX-E3-LIN-2 instance.

\begin{lemma}\label{lem:affine_opt}
For every MAX-E3-LIN-2 instance $\mathcal E$, the instance $\Pi_{\mathcal E}$ constructed in Subsection~\ref{sec:reduction} satisfies
\[
\OPT_{\mathrm{med}}(\Pi_{\mathcal E})=26q-\OPT_{\mathrm{LIN}}(\mathcal E).
\]
\end{lemma}

\begin{proof}
Note every input permutation in $\Pi_{\mathcal E}$ has length $N+M$. Then, 
\begin{align*}
    \OPT_{\mathrm{med}}(\Pi_{\mathcal E})&=4(N+M)-\max_{\tau\in\mathcal P(\Sigma\sqcup\Sigma_Z)}\sum_{\pi\in \Pi_{\mathcal E}}\LCS(\tau,\pi) &&\text{(By definition)}\\
    &=4(N+M)-\bigl(4M+22q+\OPT_{\mathrm{LIN}}(\mathcal E)\bigr)&&\text{(By Lemma~\ref{lem:anchor_block} and~\ref{lem:partition-score})}\\
    &=26q-\OPT_{\mathrm{LIN}}(\mathcal E) && \text{(Since $N=12q$, $M=2N+1$)}.
\end{align*}
\end{proof}

We are now ready to complete the proof of~\autoref{thm:four-input-hardness}.
\begin{proof}[Proof of~\autoref{thm:four-input-hardness}]
For the sake of contradiction, assume that there exists a polynomial time $(51/50 - \varepsilon)$-approximation algorithm $\mathcal A$ for the Ulam median problem for some $0 < \varepsilon < 1/50$. Given a MAX-E3-LIN-2 instance $\mathcal E$ with $q$ equations, we first construct an Ulam median instance $\Pi_\mathcal E$ (as descfribed in Subsection~\ref{sec:reduction}) and then run the approximation algorithm $\mathcal A$ to obtain a solution $\rho$. Using the affine relation in \autoref{lem:affine_opt} and the approximation guarantee, we have, for every $0 < \delta < \varepsilon$,
    \begin{itemize}
        \item If $\OPT_\textup{LIN}(\mathcal E) \geq (1-\delta)q$, then $\cost_{\Pi_\mathcal E}(\rho) \leq (51/50 - \varepsilon)\OPT_\textup{med}(\Pi_\mathcal E) < (\frac{51}2 - \delta)q$, and
        \item If $\OPT_\textup{LIN}(\mathcal E) \leq (1/2 +\delta)q$, then $\cost_{\Pi_\mathcal E}(\rho) \geq \OPT_\textup{med}(\Pi_\mathcal E) \geq (\frac{51}2 - \delta)q$.
    \end{itemize}
    That is, a $(51/50-\varepsilon)$-approximation algorithm for the Ulam median problem for any $0 < \varepsilon < 1/50$ would distinguish MAX-E3-LIN-2 instances in which at least a $1-\delta$ fraction of the equations can be satisfied from those in which at most a $1/2+\delta$ fraction of equations can be satisfied for some $\delta > 0$, contradicting \autoref{thm:hastad_lin}.
\end{proof}

Next, we use the same reduction to refute the plausibility of attaining an additive approximation; more specifically, we show the following.
\MedianAdditiveHardness*
\begin{proof}
    Assume for contradiction that there exists a polynomial-time $\varepsilon mn$-additive approximation for Ulam median with some $0 < \varepsilon \leq 1/300$. Given a MAX-E3-LIN-2 instance $\mathcal E$ with $q$ equations, we first construct an Ulam median instance $\Pi_\mathcal E$ consisting of $m=4$ input permutations and then run the additive approximation algorithm to obtain a solution $\rho$ with length $n = 36q+1$. Using the affine relation in \autoref{lem:affine_opt} and the approximation guarantee, we have, for every $0 < \delta < \varepsilon$,
    \begin{itemize}
        \item If $\OPT_\textup{LIN}(\mathcal E) \geq (1-\delta)q$, then $\cost_{\Pi_\mathcal E}(\rho) \leq \OPT_\textup{med}(\Pi_\mathcal E) +  4\varepsilon(36q+1) <  (\frac{51}{2}-\varepsilon)q$,
        \item If $\OPT_\textup{LIN}(\mathcal E) \leq (1/2 +\delta)q$, then $\cost_{\Pi_\mathcal E}(\rho) \geq \OPT_\textup{med}(\Pi_\mathcal E) > (\frac{51}{2} - \varepsilon )q$.
    \end{itemize}
    This means that an $\varepsilon mn$-additive approximation for the Ulam median problem for any $0 < \varepsilon \leq 1/300$ would distinguish MAX-E3-LIN-2 instances in which at least a $1-\delta$ fraction of the equations can be satisfied from those in which at most a $1/2+\delta$ fraction of equations can be satisfied for some $0 < \delta < \varepsilon$, contradicting \autoref{thm:hastad_lin}.
\end{proof}

\section{Hardness of Approximation for Ulam Center}
\label{sec:center-hardness}

We now reduce Ulam median to Ulam center while preserving both the number of input permutations and the optimum value. Together with Theorem~\ref{thm:four-input-hardness}, this yields the claimed hardness of Ulam center.
\CenterHardnessTheorem* 

We first describe an approximation factor-preserving reduction from the Ulam median problem to the Ulam center problem.

\begin{lemma}[Median-to-center reduction]
\label{lem:median-to-center}
There is a polynomial-time algorithm that, given an instance $\Pi=\{\pi_j\}_{j=0}^{m-1}$ of Ulam median over an alphabet $\Sigma$ of size $n$, constructs an instance $\mathcal C(\Pi)=\{\sigma_j\}_{j=0}^{m-1}$ of Ulam center consisting of $m$ permutations over an alphabet $\Sigma'$ of size $mn$ such that $\OPT_{\mathrm{ctr}}(\mathcal C(\Pi))=\OPT_{\mathrm{med}}(\Pi)$. Moreover, from any center solution for $\mathcal C(\Pi)$ of radius at most $R$, one can compute in polynomial time a permutation $\rho\in\mathcal P(\Sigma)$ with $\cost_\Pi(\rho)\leq R$.
\end{lemma}

\begin{figure}[t]
\centering
\begin{tikzpicture}[
  font=\small,
  ctrbase/.style={draw,minimum width=1.25cm,minimum height=.50cm,
                  inner sep=1pt,font=\scriptsize},
  ctrone/.style={ctrbase,fill=blue!13},
  ctrtwo/.style={ctrbase,fill=orange!20},
  ctrthree/.style={ctrbase,fill=green!16},
  ctrfour/.style={ctrbase,fill=red!13},
  ctrcandidate/.style={ctrbase,fill=black!7},
  ctrnote/.style={draw,rounded corners,fill=black!3,align=left,
                  text width=5.2cm,inner sep=5pt,font=\small}
]
  \node[font=\scriptsize] at (0,.55) {$\Sigma^{(1)}$};
  \node[font=\scriptsize] at (1.35,.55) {$\Sigma^{(2)}$};
  \node[font=\scriptsize] at (2.70,.55) {$\Sigma^{(3)}$};
  \node[font=\scriptsize] at (4.05,.55) {$\Sigma^{(4)}$};

  \node[anchor=east,font=\scriptsize] at (-.75,0) {$\sigma_1$};
  \node[ctrone] at (0,0) {$\pi_1^{(1)}$};
  \node[ctrtwo] at (1.35,0) {$\pi_2^{(2)}$};
  \node[ctrthree] at (2.70,0) {$\pi_3^{(3)}$};
  \node[ctrfour] at (4.05,0) {$\pi_4^{(4)}$};

  \node[anchor=east,font=\scriptsize] at (-.75,-.62) {$\sigma_2$};
  \node[ctrtwo] at (0,-.62) {$\pi_2^{(1)}$};
  \node[ctrthree] at (1.35,-.62) {$\pi_3^{(2)}$};
  \node[ctrfour] at (2.70,-.62) {$\pi_4^{(3)}$};
  \node[ctrone] at (4.05,-.62) {$\pi_1^{(4)}$};

  \node[anchor=east,font=\scriptsize] at (-.75,-1.24) {$\sigma_3$};
  \node[ctrthree] at (0,-1.24) {$\pi_3^{(1)}$};
  \node[ctrfour] at (1.35,-1.24) {$\pi_4^{(2)}$};
  \node[ctrone] at (2.70,-1.24) {$\pi_1^{(3)}$};
  \node[ctrtwo] at (4.05,-1.24) {$\pi_2^{(4)}$};

  \node[anchor=east,font=\scriptsize] at (-.75,-1.86) {$\sigma_4$};
  \node[ctrfour] at (0,-1.86) {$\pi_4^{(1)}$};
  \node[ctrone] at (1.35,-1.86) {$\pi_1^{(2)}$};
  \node[ctrtwo] at (2.70,-1.86) {$\pi_2^{(3)}$};
  \node[ctrthree] at (4.05,-1.86) {$\pi_3^{(4)}$};

  \draw[dashed,black!55] (-.65,-2.35)--(4.70,-2.35);
  \node[anchor=east,font=\scriptsize] at (-.75,-2.82) {$\bar\tau$};
  \node[ctrcandidate] at (0,-2.82) {$\rho_1^{(1)}$};
  \node[ctrcandidate] at (1.35,-2.82) {$\rho_2^{(2)}$};
  \node[ctrcandidate] at (2.70,-2.82) {$\rho_3^{(3)}$};
  \node[ctrcandidate] at (4.05,-2.82) {$\rho_4^{(4)}$};
\end{tikzpicture}
\caption{The median-to-center reduction for $m=4$.  Cyclic shifts
place every original input once in every row and column, preserving the
optimum.}
\label{fig:median-to-center-overview}
\end{figure}

\begin{proof}
For each $i\in[m]$, let $\Sigma^{(i)}$ be a disjoint copy of $\Sigma$. For a permutation $\pi\in\mathcal P(\Sigma)$, let $\pi^{(i)}$ denote its copy on $\Sigma^{(i)}$. The center instance is defined over the alphabet $\bigcup_{i\in[m]}\Sigma^{(i)}$ as follows. For each $j\in \{0,1,\ldots,m-1\}$, define
\[
    \sigma_j:=\pi_j^{(1)}\pi_{j+_m 1}^{(2)}\cdots\pi_{j+_m (m-1)}^{(m)},
\]
where the notation $+_m$ denotes the sum modulo $m$. Thus, every $\sigma_j$ contains one copy of every input permutation, and these copies are cyclically shifted across the $m$ blocks. The construction is illustrated in \autoref{fig:median-to-center-overview}.

We first show that a center candidate can be assumed to respect the common block order. Let $\tau$ be an arbitrary permutation of $\bigcup_{i\in[m]}\Sigma^{(i)}$, and let $\overline\tau$ be obtained by stably sorting the symbols of $\tau$ according to the block order $\Sigma^{(1)},\Sigma^{(2)},\ldots,\Sigma^{(m)}$, while preserving the relative order within each block. Every $\sigma_j$ has this same block order. Hence every common subsequence of $\tau$ and $\sigma_j$ already visits the blocks in this order, and stable sorting preserves such a subsequence. Therefore, for every $j\in\{0,1,\ldots,m-1\}$,
\[
    d_U(\overline\tau,\sigma_j)\leq d_U(\tau,\sigma_j).
\]
For each $i \in [m]$, let $\rho_i \in \mathcal{P}(\Sigma)$ denote the restriction of $\overline{\tau}$ to the $i$-th alphabet copy $\Sigma^{(i)}$, where the symbols are replaced by their counterparts in the original alphabet $\Sigma$. Since $\overline\tau$ and $\sigma_j$ have the same block order, their LCS decomposes across the blocks, and therefore
\[
    d_U(\overline\tau,\sigma_j)=\sum_{i\in[m]}d_U(\rho_i,\pi_{j+_m (i-1)}).
\]
Averaging over $j\in[m]$ gives
\begin{align*}
    \frac{1}{m}\sum_{j = 0}^{m-1}d_U(\overline\tau,\sigma_j)&=\frac{1}{m}\sum_{j = 0}^{m-1}\sum_{i\in[m]}d_U(\rho_i,\pi_{j+_m (i-1)})\\
    &=\frac{1}{m}\sum_{i\in[m]}\cost_\Pi(\rho_i)\\
    &\geq\OPT_{\mathrm{med}}(\Pi).
\end{align*}
The second equality holds because, for every fixed $i$, the (cyclic) index $j+_m (i-1)$ ranges over $\{0,1,\ldots,m-1\}$ exactly once as $j$ ranges over $\{0,1,\ldots,m-1\}$. Since
\[
\max_j d_U(\tau,\sigma_j)\geq\max_j d_U(\overline\tau,\sigma_j)\geq \frac{1}{m}\sum_j d_U(\overline\tau,\sigma_j),
\]
every center candidate $\tau$ has radius at least $\OPT_{\mathrm{med}}(\Pi)$. Hence $\OPT_{\mathrm{ctr}}(\mathcal C(\Pi))\geq\OPT_{\mathrm{med}}(\Pi).$

For the reverse inequality, let $\rho^*$ be an optimal median for $\Pi$, and consider
\[
    \tau^*:=(\rho^*)^{(1)}(\rho^*)^{(2)}\cdots(\rho^*)^{(m)}.
\]
For every $j\in \{0,1,\ldots,m-1\}$, the block decomposition gives
\[
    d_U(\tau^*,\sigma_j)=\sum_{i\in[m]}d_U(\rho^*,\pi_{j+_m (i-1)})=\sum_{\ell = 0}^{m-1}d_U(\rho^*,\pi_\ell)=\OPT_{\mathrm{med}}(\Pi).
\]
Thus, $\tau^*$ has radius $\OPT_{\mathrm{med}}(\Pi)$, and therefore $\OPT_{\mathrm{ctr}}(\mathcal C(\Pi))\leq\OPT_{\mathrm{med}}(\Pi)$. Together with the previous lower bound, this proves equality of the optimum values.

Finally, suppose that a center solution $\tau$ has radius at most $R$. Construct $\overline\tau$ and $\rho_1,\ldots,\rho_m$ as above. Since stable sorting does not increase any distance, $d_U(\overline\tau,\sigma_j)\leq R$ for every $j$. Therefore,
\[
    \frac{1}{m}\sum_{i\in[m]}\cost_\Pi(\rho_i)=\frac{1}{m}\sum_{j= 0}^{m-1}d_U(\overline\tau,\sigma_j)\leq R.
\]
Hence at least one $i\in[m]$ satisfies $\cost_\Pi(\rho_i)\leq R$. Computing all $\rho_i$ and returning one of minimum cost proves the extraction claim.
\end{proof}

\begin{proof}[Proof of \autoref{thm:center-hardness}]
Suppose that, for some $\varepsilon>0$, there is a polynomial-time $(51/50-\varepsilon)$-approximation algorithm for Ulam center with four input permutations. Given any four-input Ulam median instance $\Pi$, apply Lemma~\ref{lem:median-to-center} to construct the four-input center instance $\mathcal C(\Pi)$. Since the construction preserves the optimum value, we can compute a solution $\rho$ for the Ulam median with
\[
\cost_\Pi(\rho)\leq (51/50-\varepsilon)\OPT_{\mathrm{ctr}}(\mathcal C(\Pi)) = (51/50-\varepsilon)\OPT_{\mathrm{med}}(\Pi),
\]
contradicting~\autoref{thm:four-input-hardness}.
\end{proof}

Next, we use the same reduction to refute the plausibility of attaining an additive approximation; more specifically, we show the following.

\CenterAdditiveHardness*

\begin{proof}
    Assume for contradiction that there exists a polynomial-time $\varepsilon n$-additive approximation algorithm for the Ulam center problem with some $0 < \varepsilon \leq 1/300$. Given an Ulam median instance $\Pi$ consisting of four input permutations, each of length $\ell$, we first construct an Ulam center instance $\mathcal{C}(\Pi)$ consisting of $m=4$ permutations, each of length $n=4\ell$. We then run the additive approximation algorithm to obtain a solution $\tau$. By \autoref{lem:median-to-center}, from $\tau$, we can compute in polynomial time a solution $\rho$ for the Ulam median problem of the instance $\Pi$,  satisfying
    \[
    \mathrm{cost}_{\Pi}(\rho) \leq \mathrm{rad}_{\mathcal{C}(\Pi)}(\tau) \leq \mathrm{OPT}_{\mathrm{ctr}}(\mathcal{C}(\Pi)) + \frac{1}{300}\cdot 4\ell = \mathrm{OPT}_{\mathrm{med}}(\Pi) + \frac{1}{300}\cdot 4\ell,
    \]
    contradicting~\autoref{thm:median-additive-hardness}.
\end{proof}

\section{Conclusion}\label{sec:conclusion}

We established explicit constant-factor inapproximability for rank
aggregation under the Ulam metric.  In particular, unless $\mathrm{P}=
\mathrm{NP}$, neither Ulam median nor Ulam center admits a PTAS, even when
the instance consists of only four input permutations.  For the median
objective, the reduction combines an affine parity gadget with reversed
polarity blocks and a long anchor, yielding an exact affine relation between
the median cost and the optimum of MAX-E3-LIN-2.  For the center objective,
the cyclic block construction preserves the optimum value and permits the
lossless extraction of a median solution from any center solution.  The same
linear gap also rules out polynomial-time additive approximation schemes.

A natural next step is to narrow the substantial gap between these lower
bounds and the best known approximation algorithms.  It would be especially
interesting to obtain stronger hardness of approximation even for an arbitrary number of
inputs, or to identify additional structure that can be exploited
algorithmically.  For Ulam center, determining whether one can beat factor
$2$ when the number of inputs is part of the input remains another compelling
open direction.

\paragraph*{AI Disclosure. }The core technical idea underlying this paper was initially developed with the assistance of ChatGPT 5.6 Sol. The authors carefully examined, refined, and independently verified the resulting ideas and arguments, and rewrote to provide their own motivation, exposition, and proof structure. The final paper reflects the authors’ own understanding of the results, and the authors take full responsibility for all mathematical claims, proofs,  and references.

\bibliographystyle{plain}
\bibliography{ref}

\end{document}